\pdfoutput=1
\documentclass[12pt]{article}

\usepackage{fullpage,amsmath,amsfonts,amssymb,amsthm,epsfig,color, graphicx,
setspace}

\usepackage{float}
\usepackage{framed}
\usepackage{complexity} % computational complexity class names
\usepackage[bookmarks=false,colorlinks]{hyperref}
\hypersetup{linkcolor=red,filecolor=red,citecolor=red,urlcolor=red}

\newtheorem{theorem}{Theorem}
\newtheorem{lemma}[theorem]{Lemma}
\newtheorem{proposition}[theorem]{Proposition}

\newtheorem{claim}[theorem]{Claim}
\newtheorem{example}[theorem]{Example}

\newtheorem{definition}[theorem]{Definition}

\newcommand{\namedref}[2]{\hyperref[#2]{#1~\ref*{#2}}}

\newcommand{\figurerefb}[2]{\hyperref[#1]{Figure~\ref*{#1}#2}}

\newcommand{\equationref}[1]{\hyperref[#1]{(\ref*{#1})}}
\renewcommand{\eqref}{\equationref}
\usepackage[textsize=tiny]{todonotes}

\newcommand{\comment}[1]{}
\renewcommand{\D}{\mathcal{D}}

\newcommand{\DEBUG}[1]{}

\newcommand{\argmax}{\operatornamewithlimits{arg\,max}}

\newcommand{\trim}{ \mathtt{trim} }

\renewcommand{\setminus}{-}

\renewcommand{\ln}{\log}  %Zeyuan: I'm happy to define it the other way around as well ;) Just to make sure on the consistency

\newcommand{\vc}[1]{\mathbf{#1}}
\newcommand{\rev}{\textsc{Rev}}

\def\abs#1{\left|#1  \right|}

\begin{document}

\title{A Field Guide to Personalized Reserve Prices}

\author{
Renato Paes Leme
\and  Martin P\'{a}l
\and Sergei Vassilvitskii
}

\maketitle

\begin{abstract}
We study the question of setting and testing reserve prices in single item auctions when the bidders are not identical.  At a high level,  there are two generalizations of the standard second price auction: in the {\em lazy} version we first determine the winner, and then apply reserve prices; in the {\em eager} version we first discard the bidders not meeting their reserves, and then determine the winner among the rest. We show that the two versions have dramatically different properties: lazy reserves are easy to optimize, and A/B test in production, whereas eager reserves always lead to higher welfare, but their optimization is NP-complete, and naive A/B testing will lead to incorrect conclusions.  Despite their different characteristics, we show that the overall revenue for the two scenarios is always within a factor of 2 of each other, even in the presence of correlated bids.  Moreover, we prove that the eager auction dominates the lazy auction on revenue whenever the bidders are independent or symmetric. We complement our theoretical results with simulations on real world data that show that even suboptimally set eager reserve prices are preferred from a revenue standpoint. 
 \end{abstract}

\section{Introduction}

%Auctions are a key driver behind the multi-billion dollar online advertising
%industry. Whether it is sponsored search or display, billions of advertising
%opportunities are auctioned every day as consumers surf the web, browse their
%newsfeeds, or search for particular items. As such, online auctions have been a
%great test bed for classical auction theory, verifying predictions made in
%theoretical literature (e.g. that first pricing leads to oscillation ~\cite{EdelmanO07}),
%and also advancing the state of the art for auction design~\cite{},
%matching~\cite{}, and pricing~\cite{}.

A key part of auctions is setting the minimum price at which the seller is
willing to part with the item.  The so called reservation, or, {\em reserve},
price is critical to maximizing revenue, as proven by Myerson in his Nobel
prize winning work ~\cite{Myerson81}. In the online advertising scenario, setting the reserve price
is a non-trivial exercise---the auctions are repeated, hence agents may be
adapting their behavior to influence the learning~\cite{AminRS13,MohriM14,CesaBianchiGM15}, only a
glimpse into the buyers' valuations is known~\cite{DhangwatnotaiRY15}, and there is a
heterogeneity in the sophistication level of the bidders. On the other hand, as
observed by Celis et al.~\cite{CelisLMN14}, many of the auctions have very few
bidders (the median number of bidders is six in their dataset), thereby
exacerbating the need for good reserve prices to maximize revenue.

The fact that reserve prices are a good idea comes from Myerson's
seminal work \cite{Myerson81} which shows that if the
valuations are drawn independently and identically (iid) from a distribution satisfying a certain regularity
assumption, then the optimal auction takes the form of a second price auction with a reserve price. 
% with an uniform reserve price, where by uniform, we mean that all agents are
% subject to the same reserve. 
Myerson's result generalizes past the iid setting:
if agents are independent but not identical, the optimal auction involves 
sorting the agents by a function of their bid (called the \emph{virtual value})
and assigning the item to the agent with largest non-negative virtual value.

In practice, the optimal auction is complicated to implement, since
it involves learning the distributions of valuations in a robust enough manner to
allow the computation of virtual values. And even if this computation were
feasible, in the spirit of Wilson doctrine \cite{wilson1985}, simple detail-free
mechanisms are often preferred in complex environments in order to mitigate the
risks introduced by assumptions of the model. In the spirit of designing
simpler auctions for the non-iid setting, Hartline and Roughgarden
\cite{HartlineR09} show that if players are independent and their distribution
obey the monotone hazard rate condition, then there exist a vector of
personalized reserve prices that generates revenue which is at least half of the
optimal revenue. Both assumptions (independence and monotone-hazard-rate) are
necessary for their result.

In this work, we revisit the topic of setting personalized reserve
prices in second price auctions from the perspective of a practitioner: How to
compute them? How to make sure they are computed correctly? How to apply them in an efficient manner?  
While asking this question, we avoid (in most of our results) making any
assumptions about the shape of the bid distributions.
% and work solely with historical data from previous auctions.

{\bf Evaluation.} In practice computing a reserve price that behaves well in offline simulations 
is only the first part of the process. A successful experimental evaluation is 
key in proving that assumptions made in theory are reasonable, and do not 
lead to unintended consequences. 

Testing reserve prices is a separate research challenge in and of itself. We distinguish 
between short term studies, where the goal is to measure the immediate benefit 
of the reserve prices, and long term studies that intend to capture the strategic interactions between the bidders and  
auctioneer. Here we focus on the short term studies, and observe that testing 
personalized reserve prices, even in this, relatively simple, setting, is non-trivial, and 
can lead to incorrect conclusions. 

\subsection{Our Contributions}
We show that the problem of computing and applying personalized 
reserve prices is nuanced, and
has paradoxical behavior, particularly when testing the efficacy of
reserves in an A/B test. Our empirical evaluations show that the paradoxical
behavior is not limited to theory, but does occur in practice as well.

We begin by describing two different approaches to applying personalized reserve
prices: {\em lazy} and {\em eager}, and show that while sometimes one approach
dominates another on revenue, they are always within a factor of two of each
other (Theorem \ref{thm:lazy_eager_bound}).

We then identify two mildly restrictive settings (those of symmetric bidders, and independent bidders) 
and show that in these situations optimal eager reserve prices always yield more revenue (Theorems \ref{claim:eager_better_symmetry} and \ref{theorem:eager_better_independence}).

Turning to computational issues in setting the optimal lazy and eager
reserve prices, we show that a simple nearly linear time algorithm can compute
the optimal lazy reserves from previous history, while computing the optimal eager
reserves is NP-hard (Propositions \ref{prop:linear_lazy} and \ref{prop:nphard}). 

We show that naive A/B testing of eager reserve
prices  {\em always} leads to a drop in revenue, even when the reserve
prices are set correctly! (Theorem \ref{thm:revenue-testing-eager}) We observe
that lazy reserves do not suffer from this problem, and behave in an intuitive
manner.

Finally, we present an empirical evaluation of our findings on
real world data, and show that  the performance of the algorithms is much better
than the large approximation factors guaranteed by the theory (Section \ref{sec:empirical}). 

\section{Preliminaries}
We consider the standard setting of single item auctions.  Let $A =
\{1,2,\ldots, n\}$ be the set of agents interested in the item. Each agent $i$
has a value $v_i$ for the item, and submits a bid $b_i$ to the auctioneer. Given
a vector of bids ${\bf{b}} = (b_1, b_2, \ldots, b_n)$, we denote by $b^{(1)}$ the highest bid, and $b^{(2)}$ the second highest bid. 
We will denote 
by ${\bf r} = (r_1, r_2, \ldots, r_n)$ the vector of personalized reserve prices.

We assume that the valuation of each agent is drawn independently from an
unknown distribution with CDF $F$ and PDF $f$.  To ease the exposition we assume
that the distributions $F$ are regular, in that the virtual value function
$\phi(v) = v - \frac{1 - F(v)}{f(v)}$ is monotone non-decreasing. 

In the case when all of the distributions are identical and known to the seller,
$F_1 = F_2 = \ldots = F_n$, Myerson proved the following characterization of the
optimal auction: \begin{enumerate}\itemsep=0in \item Collect bids $b_1, b_2,
      \ldots, b_n$.  \item Discard all bids that are below $\phi^{-1}(0)$.
  \item Allocate the item to the agent with the highest bid, and charge her the
maximum of $\phi^{-1}(0)$ and $b_{(2)}$.  \end{enumerate}

Note that $\phi^{-1}(0)$ acts as a reserve price for the auction: this is the
minimum bid any agent must submit to win, and also acts as a lower bound on the
revenue to the seller.  

In this work our focus is on the non-identical setting. In this case Myerson
proved that the optimal auction is: \begin{enumerate}\itemsep=0in \item Collect
    bids $b_1, b_2, \ldots, b_n$.  \item For each agent, discard bid $b_i$ if it
      is lower than $\phi^{-1}_i(0)$.  \item Allocate the item to the agent with
        the highest virtual value, $\phi_i(b_i)$, and charge her the maximum of
        $\phi_i^{-1}(0)$ and $\phi_i^{-1}(\phi_j(b_j))$, where $j$ is the agent
        with the second highest virtual value.  \end{enumerate}

There are two major differences from the identical setting. First, instead of
having a universal reserve price $\phi^{-1}(0)$, each agent now has a
personalized reserve price  $\phi_i^{-1}(0)$.  Second, the winner is determined
as the agent whose bid has the highest virtual value, not the one who has the highest bid. 

While the first auction is easily implemented in practice, the latter is much
more problematic. First, the virtual value functions depend critically on the
value distributions, which themselves are not always known and must be
estimated. Second, the auction is counterintuitive to outsiders, as the agent
with the highest bid does not always win the item. 

\subsection{Personalized Reserve Prices} To combat the potential bid inversion
that comes with ordering by virtual values,  Hartline and Roughgarden
\cite{HartlineR09} proposed keeping the personalized reserve prices aspect
of the optimal auction, but ordering items by bid instead. They show that with 
monopoly reserves $\phi_i^{-1}(0)$, this auction yields a $2$-approximation to the revenue
of the optimal auction when agents are independent and follow the monotone hazard
rate condition.

This natural approach leads to two possible flavors of second price auctions,
which were first introduced by Dhangwatnotai et al.~\cite{DhangwatnotaiRY15}.
Informally, in the {\em eager} regime, we first discard all of the bids that do
not meet their personalized reserve prices and then run the second price auction
on the remaining bids. In the {\em lazy} regime, we always try to allocate the
item to the agent with highest bid. If her bid is below the reserve price the
good goes unallocated, otherwise the agent is charged the maximum between her
reserve price and the second highest bid. We describe these two auctions more
formally in Section \ref{sec:lazyeagerauctions}.

\subsection{Estimating Value Distributions} Myerson's theoretical analysis
requires us to know the cumulative density function $F$, as well as the PDF,
$f$, to compute the virtual value function $\phi(\cdot)$. In practice this is a
tall order. In the online advertising context, the auctions are repeated, and
thus we can observe a number of draws from the distribution. For example, we
observe the sequence of bids for agent $i$: $b^1_i, b^2_i, b^3_i, \ldots$, from
which we can compute empirical estimates  $\hat{F}$ and $\hat{f}$ for  $F$ and
$f$.
%Previously,  Dhangwatnotai et al.~\cite{DhangwatnotaiRY15} has shown that
%these are not restrictive assumptions, and reasonable reserve prices can be
%learned even when we estimate $F$ by a single sample from the distribution. 

In this work we focus on the computational complexity of computing optimal
personalized reserve prices from the previous bids. In particular we show that
the optimization question for eager and lazy auctions has very different
profiles, one being solvable in polynomial time, and the other being
NP-complete.

\section{Lazy and Eager Auctions}\label{sec:lazyeagerauctions}
Insisting that advertisers are ranked by bid leads to two flavors of second price auctions. 
As before, let ${\bf b} = (b_1, b_2, \ldots, b_n)$ be the bids submitted to the auctioneer, and assume
without loss of generality that $b_1 \geq b_2 \geq b_3 \ldots b_n$.  Let ${\bf r} =  (r_1, r_2, \ldots, r_n)$ be 
the vector of reserve prices, with reserve price $r_i$ applying to bidder $i$. 

Following the work of Dhangwatnotai, Roughgarden and Yan \cite{DhangwatnotaiRY15}, we define second 
price auctions with {\em lazy} and {\em eager} reserves. 

{\bf Lazy Reserves}:
\begin{itemize} \itemsep=0in
\item If $b_1 < r_1$ then there is no winner and the item goes unsold. 
\item if $b_1 \geq r_1$, allocate the item to bidder $1$ and charge her $\max(r_1, b_2)$. 
\end{itemize}

{\bf Eager Reserves}:
\begin{itemize}\itemsep=0in
\item Let $S = \{i : b_i \geq r_i\}$ be the set of bidders who bid above reserve. 
\item Let $j$ be bidder with the highest bid in $S$ (ties broken by the original ordering). 
\item Allocate the item to bidder $j$ and charge her the maximum of her reserve and second highest bid 
in $S$: $\max(r_j, \max_{i \in S \setminus \{j\}} b_i)$.
\end{itemize}

Observe that when the reserve prices are identical,  $r_1 = r_2 = \ldots = r_n$, then both of these
approaches implement the standard second price auction.  We denote by $\rev_L(\vc{b}; \vc{r})$ and $\rev_E(\vc{b}; \vc{r})$ the 
revenue obtained by running the lazy and eager auctions on the same set of bids and reserve prices. 

Both versions of the auction are incentive compatible and individually rational,
since, whenever player $i$ wins with bid $b_i$, she still wins with all bids
$b'_i \geq b_i$. Moreover, in either case, the payment of $i$ correspond to her
critical bid, i.e., the infimum of the bids for which he wins.

\subsection{Example}
To demonstrate the difference between the two approaches and show why in general the revenues 
are incomparable, consider the following two examples.

\begin{example}[Eager Dominates Lazy]
Consider an auction with three bidders $A$, $B$, and $C$, who bid $7$, $5$, and $3$ respectively. Suppose the vector
of reserve prices is $8$, $1$ and $2$. The auction with lazy reserves tries to allocate the item to $A$ since she 
has the highest bid. However, since her bid is lower than her reserve price, the item goes unallocated, and the seller collects no revenue. 
The auction with eager reserves first filters out $A$, allocates  the item to $B$ and collects revenue of $3$. 
\end{example}

\begin{example}[Lazy Dominates Eager]
As above let $A$, $B$, and $C$ be the three bidders bidding $7$, $5$, and $3$. Suppose the vector of reserve prices is
$2$, $6$ and $1$. The auction with lazy reserves allocates the item to $A$ since her bid is above her reserve and charges her 
$5$ (the second highest bid). On the other hand, the auction with eager reserves first removes $B$ from consideration (since his bid 
is below reserve). $A$ still wins the item, but is charged only $3$.
\end{example}

Intuitively, eager reserve prices lead to higher revenue when the highest bid is priced out; whereas lazy reserve prices lead to higher revenue
when the second highest bidder is priced out.

\section{Comparing Lazy and Eager Reserve Prices}

The second price auctions with eager and lazy reserves are  identical
when all agents are subject to the same reserve price. If the reserve prices are
personalized, however, the outcomes of the auctions can be very different.
In terms of welfare, it is easy to see that the
auction with eager reserves always dominates the auction with lazy reserves.
The auction with lazy reserves allocates the item only if the highest
bidder is above her reserve price. In such cases, the auction with eager reserves
also allocates the item. This fact alone makes a second price auction with eager
reserves more attractive to sellers that care about match rate.

As we saw in examples in Section \ref{sec:lazyeagerauctions}, in general the revenue 
gains due to lazy and eager reserve prices are incomparable. In this section we provide a 
tighter characterization, and identify broad classes where one mechanism dominates another. 

We begin with Theorem \ref{thm:lazy_eager_bound} and prove that for any (possibly correlated)
bid distribution no auction generates more than twice the revenue of the other, 
and give two examples to show this bound is asymptotically tight.

Then, in Section \ref{sec:restricted} we show that the revenue of the optimal
eager mechanism dominates that of the optimal lazy mechanism whenever either
(i) the bidders are symmetric (the joint bid distribution doesn't change when 
bidders are permuted), or, (ii) the bidders' bids are drawn from independent
(not necessarily identical) distributions.

\begin{theorem}\label{thm:lazy_eager_bound}
Let $\mathcal{D}$ be any distribution over bid vectors $\vc{b}$.
Also, let $\rev_L(\vc{b}; \vc{r})$ and $\rev_E(\vc{b}; \vc{r})$ denote the
revenue of the lazy and eager auctions under bid vectors $\vc{b}$ and
personalized reserves $\vc{r}$. Let:
$$\vc{r}_L^* = \argmax_{\vc{r}} \E_{\vc{b} \sim \mathcal{D}}[ \rev_L(\vc{b},
\vc{r})],$$
$$\vc{r}_E^* = \argmax_{\vc{r}} \E_{\vc{b} \sim \mathcal{D}}[ \rev_L(\vc{b},
\vc{r})].$$ Then
$$\E_{\vc{b} \sim \mathcal{D}}[ \rev_L(\vc{b}, \vc{r}^*_L)] \leq 2 \cdot \E_{\vc{b}
\sim \mathcal{D}}[ \rev_E(\vc{b}, \vc{r}^*_E)] ,$$
$$\E_{\vc{b} \sim \mathcal{D}}[ \rev_E(\vc{b}, \vc{r}^*_E)] \leq 2 \cdot \E_{\vc{b}
\sim \mathcal{D}}[ \rev_L(\vc{b}, \vc{r}^*_L)] .$$
\end{theorem}

\begin{proof}
For the first inequality, consider running the eager second price auction with
$\vc{r}^*_L$. For any bid vector $\vc{b}$, if the highest player is below
the reserve then: $\rev_L(\vc{b}, \vc{r}) = 0 \leq \rev_E(\vc{b}, \vc{r})$. If
the highest player is above the reserve and her payment in the lazy auction is
her reserve, then both auctions generate the same revenue. The final case is the
case where the highest bidder is above the reserve and the payment in the lazy
auction is the second highest bid. In this case, the revenue of the eager
auction might be lower if the second highest bidder is below her reserve. For
that bid vector, however, the revenue of the lazy auction is the second highest
bid, which is equal to the revenue of the auction with no reserve prices.
Therefore:
$$\begin{aligned}
\rev_L(\vc{b}, \vc{r}^*_L) & \leq \max[ \rev_E(\vc{b}, \vc{r}_L^*),
\rev_E(\vc{b}, \vc{0})] \\ & \leq \rev_E(\vc{b}, \vc{r}_L^*) + \rev_E(\vc{b},
\vc{0}) \end{aligned}$$
Taking expectations over $\vc{b}$ we obtain the first inequality:
$$\begin{aligned}
  \E[\rev_L(\vc{b}, \vc{r}^*_L)] & \leq \E[\rev_E(\vc{b}, \vc{r}_L^*) +
\rev_E(\vc{b},    \vc{0})] \\ & \leq 2 \cdot \E[ \rev_E(\vc{b}, \vc{r}_E^*)]
\end{aligned}$$

For the second inequality, consider running a lazy second price auction with
$\vc{r}_E^*$.  If for any bid vector $\vc{b}$, the highest player is above
her reserve, the lazy auction is guaranteed to generate more revenue than the
eager auction. If not, then the eager auction generates at most revenue equal
to the bid of the second highest bidder (since revenue is dominated by welfare),
which is the revenue of the second price auction with no reserves. Therefore:
$$\begin{aligned}
\rev_E(\vc{b}, \vc{r}^*_E) &\leq \max[ \rev_L(\vc{b}, \vc{r}_E^*),
\rev_L(\vc{b}, \vc{0})] \\ & \leq \rev_L(\vc{b}, \vc{r}_E^*) + \rev_L(\vc{b},
\vc{0}) 
\end{aligned}$$
We obtain the second inequality by again taking expectations over $\vc{b}$:
$$\begin{aligned}
  \E[\rev_E(\vc{b}, \vc{r}^*_E)] &\leq \E[\rev_L(\vc{b}, \vc{r}_E^*) + \rev_L(\vc{b},
\vc{0})] \\ & \leq 2 \cdot \E[\rev_L(\vc{b}, \vc{r}_L^*)]
\end{aligned}$$
\end{proof}

%If the valuations are drawn independently (but not identically)
%from distributions $F_i$ satisfying the monotone
%hazard rate condition, Hartline and Roughgarden \cite{HartlineR09} showed that
%the if $r_i^M$ is the Myerson's reserve price for distribution $F_i$,
%i.e. $r_i^M = \argmax_r r (1-F(r))$, and $\vc{r}^M =
%(r_1^M, \hdots, r_n^M)$, then $\E[\rev_E(\vc{b}, \vc{r}^M)]$ is at least half of the
%revenue of the optimal auction. If the distributions are not MHR or if they are
%not independent, \cite{HartlineR09} show that there is no constant factor bound
%between the revenue of the optimal auction and the revenue of the optimal (lazy
%or eager) second price auction with personalized reserves.

\subsection{Lower Bound Examples}

The following examples complement  the bounds in Theorem 
\ref{thm:lazy_eager_bound}. First, we show an example where revenue
from the eager mechanism is almost twice the revenue from the best lazy 
mechanism.

\begin{example}\label{claim:eager_better}
There is an instance with $n$ bidders with valuations drawn independently from 
the same distribution, where the best eager mechanism generates $2-o(1)$ 
times more revenue than the best lazy mechanism.
\end{example}

\begin{proof}

Consider $n$ identical bidders, with each bidder choosing to bid $b_i=n$ with
probability $1/n^2$ and with probability $1-1/n^2$ bidding $b_i=1$.
(Each bid is drawn independently. To break ties, each bid
is perturbed by adding a noise term drawn independently from $[0,\epsilon]$
for infinitesimally small $\epsilon)$).

With probability $1-1/n+O(n^{-2})$, all bidder valuations will be 
$1$, and with probability $1/n - O(n^{-2})$, some bidder will have a valuation of $n$. 
Thus, the optimal welfare given these valuations is $2-O(n^{-2})$.

We claim that a lazy pricing mechanism extracts at most $1$ revenue. By symmetry
each buyer is equally likely to be the winner. Now consider the reserve price for an 
individual bidder $i$. 

\begin{itemize}
\item If $r_i<1+\epsilon$, revenue will be bounded by $1+\epsilon$ unless both
bidder $i$ and another bid bid high (probability of $O(n^{-2})$), in which case revenue is bounded by $n$. 
The expected revenue is then $1 + o(1)$.

\item If $r_i>1+\epsilon$, revenue is only earned if $b_i = n$, and is bounded by $n$. Since the probability of a high 
bid is less than $n^{-1}$, the expected revenue is bounded by $1$. 
\end{itemize}

On the other hand, we claim that an eager mechanism can earn $2-o(1)$ revenue.
This can be done by imposing a high reserve price $r_i = n$ for 
$i=2,\dots,n$ (all buyers except one) and a low reserve price $r_1 = 1$ on the 
remaining bidder. With this setup:

\begin{itemize}
\item With probability $1/n-o(n^{-2})$, one of bidders $2,\dots,n$ will submit a 
high bid, generating revenue $n-O(1)$.
\item With the remaining probability, the auction is guaranteed to clear 
because $b_1 \ge r_1 = 1$ and will generate at least $1$ in revenue.
\end{itemize}
The overall expected revenue is thus 
$\frac{1}{n}\cdot (n-o(1)) + (1-\frac{1}{n})\cdot 1 = 2-o(1)$.
\end{proof}

Next, we show an example where the lazy mechanism generates twice the revenue 
of the best eager mechanism. 
This example requires both that bidders are not symmetric,
and that their bids are correlated.

\begin{example}\label{claim:lazy_better}
There is an instance with two correlated heterogenous bidders where a lazy
mechanism generates $2-o(1)$ times higher revenue than the best eager mechanism.
\end{example}

\begin{proof}
Let $M$ be a sufficiently large constant, and consider the following joint distribution of bids of two bidders.
\begin{itemize}
  \item With probability $\frac{\ln M}{M}$, 
    bidder $1$ bids $0$ while bidder $2$ bids $b_2=M$.
  \item With probability $1-\frac{\ln M}{M}$, first bidder's bid $b_1$ is 
    drawn from a truncated equal revenue distribution $F(b) = 1-\frac{1}{b}$ for
    $b \in [1,M)$ and $F(M) = 1$ and the second 
    bidder's bid is set to 
    $b_2 = (1-\epsilon) b_1$ for some arbitrarily small $\epsilon$.
\end{itemize}

Lazy reserve prices $r_1 = 0$, $r_2 = M$ extract the full
surplus as $\epsilon \rightarrow 0$, achieving expected revenue of 
\begin{align*}
&\frac{\ln M}{M} \cdot M + \left(1 - \frac{\ln M}{M}\right) (1-\epsilon) \left( \int_1^M b 
 f(b) db + \frac{M}{M} \right) \\
&=\ln M + \left(1 - \frac{\ln M}{M}\right) (1-\epsilon) \cdot \left(\ln M + 1\right) \\ 
&= (2 - o(1)) \ln M.
\end{align*}

Suppose the eager mechanism imposes reserves of $r_1$ and $r_2$. With probability $\ln M/M$ the revenue is $r_2$. For the remaining case, we condition on whether the second bidder bid above reserve. In the case $b_2 < r_2$,  bidder 1 is competing against his reserve price, $r_1$, generating revenue of at most:
$$r_1 (1 - F(r_1)) \geq r_1 \left[1 - \left(1 - \frac{1}{r_1} \right) \right]  = 1.$$

If, on the other hand, when $b_2 \geq r_2$, the total revenue is bounded by:
$$\int_{r_2}^{M} b f(b) db = \ln M - \ln r_2 = \ln (M/r_2)$$
Putting these together the overall revenue is at most :
\begin{align*}
&r_2 \frac{\ln M}{M} + \left(1 - \frac{\ln M}{M}\right) \cdot \left(1 + \ln
\frac{M}{r_2} \right) \\
&\leq 1 + r_2 \frac{\ln M}{M} +  \ln M - \log r_2
\end{align*}
This function is convex in $r_2$, thus achieving its maximum at the endpoints of the interval. It obtains its maximum value of $(1+ o(1)) \ln M $ at $r_2 = 1$. 
\end{proof}

\subsection{Restricted Settings}
\label{sec:restricted}
Above we proved that while the revenue from lazy and eager auctions is always within 
a factor of two of each other, an unconditional bound is impossible in  general. 
Here we consider two restricted settings, first of symmetric bidders, and then of independent
bidders. In both cases we show that eager auctions dominate lazy auctions.\\

{\bf Symmetric Bidders.} We say that bidders are symmetric (sometimes also
called exchangeable)
if  the bid distribution is
invariant under permutations. Formally, for every permutation $\pi:[n] \rightarrow [n]$
let $\vc{b}^\pi$ be the vector $(b_{\pi(1)}, b_{\pi(2)}, \hdots, b_{\pi(n)})$.
A distribution is symmetric if the distribution of $\vc{b}^\pi$ is the same
as the distribution of $\vc{b}$. Notice that iid implies symmetry, but
symmetry is more general---for example, consider the bid distribution obtained by
choosing in each time one buyer at random and letting her bid $H$ and
letting every other buyer bid $L$. The distribution is clearly symmetric,
but it is not independent (since there is always exactly one buyer bidding $H$).

\begin{theorem}\label{claim:eager_better_symmetry}
  If bidders are symmetric the optimal eager mechanism yields at least as much
revenue as the optimal lazy mechanism.
\end{theorem}

\begin{proof}
Proposition \ref{prop:linear_lazy} tells us that the optimal lazy reserve price
for bidder $i$ is a function of the joint distribution of winning bids and prices
conditioned on bidder $i$ winning. If bidders are symmetric, the optimal
lazy reserve prices are the same for all bidders; $r_1 = r_2 = \dots = r_n$.
If all reserve prices are the same, lazy and eager mechanisms behave identically.
(It is however possible that a different reserve price vector yields higher 
revenue for the eager mechanism, as in Example \ref{claim:eager_better}.)
\end{proof}

{\bf Independent Bidders.} We show that if each bidder draws her value independently,
 auctions with eager reserves are always at least as good auctions with lazy reserves. 

\begin{theorem}\label{theorem:eager_better_independence}
If bidder valuations are drawn independently 
(not necessarily from identical distributions),
the optimal eager mechanism yields at least as much
revenue as the optimal lazy mechanism.
\end{theorem}

To prove the theorem we will exhibit a method that, given any vector $\vc{r}$ of
lazy reserve prices, produces a vector $\vc{r}_E$ such that $\E[\rev_E(\vc{b};
\vc{r}_E)] \geq \E[\rev_L(\vc{b};\vc{r})]$. The construction will rely heavily 
on the independence of the bids. A key concept in the proof is that of
trimmed distributions:

\begin{definition}[trim]
Consider a non-negative random variable $X$ distributed according to 
some distribution $\D$, and a real number $r\ge 0$. 
We use $\trim(\D)$ to denote the distribution of the 
random variable $X'=X \cdot \vc{1} \{X\ge r\}$, i.e.. random variable that is equal to 
$X$ if $X\ge R$ and is zero if $X<r$. 
\end{definition}

This definition allows us to define the main lemma. In the following proofs, we
will abbreviate $\E_{\vc{b} \sim \D}[\rev(\vc{b}, \vc{r})]$ by $\rev(\D,
\vc{r})$.

\begin{lemma}\label{lemma:trimlemma}
  Given independent bid distributions $\D_i$, $\D= \D_1 \times \hdots \D_n$
  and a reserve price vector $\vc{r}$, there exists a vector of reserve
  prices $\vc{r}'$ such that:
  $$\rev_L(\vc{\D}, \vc{r}) \leq \rev_L(\vc{\D}', \vc{r}')$$
  where $\D' = \D'_1 \times \hdots \times \D'_n$ and $\D'_i = \trim(\D_i; r'_i)$.
\end{lemma}

First we show how to use Lemma \ref{lemma:trimlemma} to prove Theorem
\ref{theorem:eager_better_independence}:

\begin{proof}[of Theorem \ref{theorem:eager_better_independence}]
  Let $\vc{r}'$ and $\D'$ be as in Lem\-ma \ref{lemma:trimlemma}. Notice that for
  $\vc{b} \sim \D'$, $\Pr[0 < b_i < r_i] = 0$. Since a bidder is never blocked by the 
  reserve, the revenue in both the
  eager and lazy auctions is the maximum of the second highest bid
  and the reserve of the highest bidder. This implies that $\rev_E(\vc{\D}', \vc{r}') =
  \rev_L(\vc{\D}', \vc{r}')$. Since $\D$ stochastically dominates $\D'$, there is a
  distribution on pairs of vectors $(\vc{b'}, \vc{b})$ such that the marginals
  are $\D'$ and
  $\D$ and $b'_i \leq b_i$ for all $i$ for every realization of the random
  variables. Hence, $\rev_E(\vc{b'};\vc{r}')]
  \leq \rev_E(\vc{b};\vc{r}')$. Taking expectations we conclude that
  $$\rev_E(\vc{\D}', \vc{r}') \leq \rev_E(\vc{\D},
      \vc{r}').$$
  
  Putting this together with the inequality from Lemma \ref{lemma:trimlemma}:
  $$\begin{aligned}
    \rev_L(\vc{\D}, \vc{r})  \leq \rev_L(\vc{\D}', \vc{r}') = \rev_E(\vc{\D}',
    \vc{r}')  \leq \rev_E(\vc{\D},\vc{r}')
\end{aligned}$$
\end{proof}

Now, all is left to do is to prove Lemma \ref{lemma:trimlemma}:

\begin{proof}[of Lemma \ref{lemma:trimlemma}]
  Assume that the bidders are sorted such that $r_1 \geq r_2 \geq \hdots \geq
  r_n$. We will define an algorithmic procedure that iterates through bidders
  $1$ to $n$ and at each iteration $i$, trims the distribution of the $i$-th
  bidder and possibly increases the reserves of agents $j > i$. It is useful to
  think of $\vc{r}$ and $\D = \D_1 \times \hdots \times D_n$ as variables that
  are updated in the course of the procedure.

  The procedure will keep the following invariants: (i) $\rev_L(\D, \vc{r})$
  cannot decrease; (ii) for all bidders already processed, their distribution is
  trimmed at their reserve, i.e., $\Pr[0 < b_i < r_i] = 0$. (iii) the reserve
  prices will continue to be sorted.

  Now, we are ready to describe each iteration. When we process bidder $i$, we
  perform the following procedure:

  \begin{framed}
  \noindent Choose $x < r_i$ such that:
  $$\E_{\D}[\rev_L(\vc{b};\vc{r}) \vert b_i = x] \geq \E_{
  \D}[\rev_L(\vc{b};\vc{r}) \vert b_i < r_i] $$
  Set $\D_i = \trim(\D_i, r_i)$ and $r_j = \max(r_j, x)$ for $j > i$.
  \end{framed}
%Argue that x exists?
  Clearly we maintain invariants (ii) and (iii). Now, we only need to argue that
  invariant (i) is also maintained. It is convenient to
  write $\rev_L(\vc{b}, \vc{r}) = \sum_j \rev^j_L(\vc{b}, \vc{r})$ where
  $\rev^j_L(\vc{b}, \vc{r})$ is the revenue obtained from bidder $j$.

  First notice that $\E_{\D}[\rev_L^i(\vc{b}; \vc{r})] = \E_{\D'}[\rev_L^i(\vc{b};
  \vc{r}')]$ since the lazy auction just extracts revenue from $i$ when she is
      the highest bidder and above her reserve (and those events are unaffected
      by trimming). Also, in the lazy auction, the reserves on bidders other than
      the highest bidder do not affect the outcome.

  For all other $j \neq i$, notice that $\E_{\D}[\rev_L^j(\vc{b}; \vc{r}) \vert
  b_i \geq r_i] = \E_{\D'}[\rev_L^j(\vc{b}; \vc{r}') \vert b_i \geq r_i]$. 
  Conditioned on $b_i > r_i$ the only thing changing in the two scenarios is
  that the reserve price of $j$ is now $r'_j = \max(r_j, x)$; since $x <
  r_i$, it cannot be binding as bidder $i$ is bidding $b_i > r_i$. (Notice that for
  bidders $j < i$, their reserve price was already at least $r_i$, so $r'_j =
  r_j = \max(r_j, x)$ since $x < r_i < r_j$.)

  Finally, $\E_{\D'}[\rev_L^j(\vc{b}; \vc{r}') \vert b_i < r_i] =
  E_{\D}[\rev_L^j(\vc{b}; \vc{r}) \vert b_i = x]$, since in the first case $i$
  must be bidding zero and $j$ is subject to reserve $r_j = \max(r_j, x)$, but
  this reserve can be implemented by having $i$ bid $x$.

  Combining all of the expressions, we get:
  $$\begin{aligned}
    &\E_{\D'}[\rev_L(\vc{b}, \vc{r'})] \\ & = \E_{\D'}[\rev_L^i(\vc{b},
  \vc{r'})] \\ & \quad +  \sum_{j \neq i }\E_{\D'}[\rev_L^j(\vc{b},\vc{r}')
    \vert b_i \geq
  r_i] \Pr(b_i \geq r_i) \\ & \quad + \sum_{j \neq i
  }\E_{\D'}[\rev_L^j(\vc{b},\vc{r}') \vert b_i <
  r_i] \Pr(b_i < r_i) \\
  & =  \E_{\D}[\rev_L^i(\vc{b}, \vc{r})]  \\ & \quad +  \sum_{j \neq i
}\E_{\D}[\rev_L^j(\vc{b},\vc{r}) \vert b_i \geq
  r_i] \Pr(b_i \geq r_i) \\ & \quad + \sum_{j \neq i
}\E_{\D}[\rev_L^j(\vc{b},\vc{r}) \vert b_i =  x] \Pr(b_i <r_i) \\
  \end{aligned}$$

  Since $\E_{\D}[\rev_L^i(\vc{b}; \vc{r}) \vert b_i < r_i] = 0$, $x$ was picked
  such that 
  $$\sum_{j \neq i}\E_{\D}[\rev_L^j(\vc{b};\vc{r}) \vert b_i = x] \geq \sum_{j
  \neq i}\E_{
  \D}[\rev_L^j(\vc{b};\vc{r}) \vert b_i < r_i] $$
  Plugging this into the last expression, we get that:
  $$\E_{\D'}[\rev_L(\vc{b}, \vc{r'})] \geq \E_{\D}[\rev_L(\vc{b}, \vc{r})]$$
\end{proof}

\section{Computing Optimal Reserves}
\label{sec:complexity}
In this section we investigate the computational complexity of computing 
the optimum reserve prices in the eager and lazy settings. We assume that 
the input to the problem is given as a set of bids submitted to previous auctions. 
Since the number of auctions run daily is extremely large, it is important 
for this procedure to be linear, or nearly linear in the size of the input. 

We prove that from a computational perspective, lazy and eager
auctions are vastly different. Computing the optimal vector of lazy reserves
$\vc{r}_L^*$ can be done in linear time. On the other hand, computing the
optimal vector of eager reserves $\vc{r}^*_E$  is NP-hard.

\begin{proposition}\label{prop:linear_lazy}
The optimal vector of lazy reserves $\vc{r}_L^*$ can be computed in nearly linear time
in the size of input logs.
\end{proposition}

\begin{proof}
Given bids $\{b_{i,t}\}$, for each query $t$, let $w_t$ correspond to the agent
that would win if no reserve prices were set. For any given vector of reserves
$\vc{r}$, either $w_t$ wins query $t$ or no agent wins. Let $Q_i$ be the queries
for which $w_t = i$. Then we can write the revenue for the vector $\vc{r}$ of
reserves as:
$$\rev_L(\vc{b};\vc{r}) = \sum_i \sum_{t \in Q_i} \vc{1} \{ b_{t}^{(1)} \geq r_i \}
\cdot \max(r_i, b_t^{(2)})$$
where $b_t^{(1)}$ $b_t^{(2)}$ are respectively the highest and second highest
bid for query $t$. The previous expression shows that for lazy reserves, the
problem of computing reserve prices can be decoupled for every $i$. Also,
the optimal reserve price should be of
the form $b_t^{(2)}$ or $b_t^{(1)}$ for some $t$. If not, we can increase it to
the next point and the revenue can only increase.

This observation gives  an algorithm for computing the optimal vector of reserves with running time $O(\sum_i
\abs{Q_i}^2)$. In order to turn it
into a nearly linear time algorithm, notice that if we sort the bids
appropriately, we can compute the revenue for setting each $r_i = b_t^{(1)}$
and $r_i = b_t^{(2)}$ in constant time.

In order to do so, construct an array with all the bids $b_t^{(1)}$ and
$b_t^{(2)}$ for $t \in Q_i$ and annotate each entry of weather it is a highest
bid or a second highest bid. If $n_i = \abs{Q_i}$, this is an array of $2 n_i$
numbers. Sort the array by bids in increasing order,
which takes $O(n_i \log n_i)$, and call its elements $r_1 \leq r_2 \leq \hdots
r_{2n_i}$. Assume for simplicity that all elements are distinct. Given a certain
$r_j$, we can write the revenue associated with it $R(r_j) = \sum_{t \in Q_i}
\vc{1} \{ b_{t}^{(1)} \geq r_j \} \cdot \max(r_j, b_t^{(2)})$ as $r_j \cdot k_j
+ s_j$ where $s_j = \sum_{t \in Q_i} b_t^{(2)} \cdot \vc{1} \{ r_j \leq
b_t^{(2)} \}$.

If we show how to compute $(s_{j+1},k_{j+1})$ from $(s_j,k_j)$ in constant time,
we have a nearly-linear time algorithm. Doing it is easy. This can be done in
two cases:
\begin{itemize}
\item $r_j = b_t^{(1)}$ for some $t$. Therefore increasing the reserve past $b_t^{(1)}$ will
make query $t$ to be unallocated. All other queries are unaffected. Set $k_{j+1}
= k_j - 1$ and $s_{j+1} = s_j$.
\item $r_j = b_t^{(2)}$ for some $t$. Therefore increasing the reserve past $b_t^{(2)}$
will cause the reserve price to bind for query $t$, instead of the second
highest bid. Update: $k_{j+1} = k_j + 1$, $s_{j+1} = s_j -r_j$.
\end{itemize}
Since we can reconstruct the revenue for each value in the array, we can choose
the optimal vector of reserves in time $O(\sum_i \abs{Q_i} \cdot \log \abs{Q_i})$.
It is not hard to see that this algorithm still works if values of the array are
repeated.
\end{proof}

\begin{proposition}\label{prop:nphard}
Computing the optimal vector of eager reserves $\vc{r}$ given the bids
from a set of previous auctions is NP-hard. 
\end{proposition}

\begin{proof}
We give a reduction to the independent set problem. Let $L$ and $H$ be constants that we will chose later. 
Given a graph $G =
(V,E)$, we map the independent set problem on this graph to the following
instance of the reserve price problem with eager reserves: consider $\abs{V}$
agents and $\abs{E} + \abs{V}$ queries. 
For each edge $e =
(u,v)$ consider queries where $b_u = L$, $b_v = L$ and all other agents bid
zero. And for each node $u$, consider queries with $b_u = H$ and all other
agents bid zero. We select $L$ and $H$ such that $L < H < 2L$.

Clearly, for the optimal vector of reserves, $r_i \in \{L,H\}$. Also, notice
that in the optimal solution, the set of nodes $I = \{u \in V; r_u = H\}$ must form
an independent set. Indeed, if there is an edge $e = (u,v)$ with $r_u = r_v =
H$, then there is zero revenue from the queries corresponding to edges $e$. If
we switch either $u$ or $v$ to have reserve $L$, then we gain $L$ revenue from
edge $e$ and lose $H-L$ from node $u$. Since $L > H-L$, this is a profitable
deviation.

The revenue associated with setting $u \in I$ to $H$ and other nodes to $L$ is
given by $$L\cdot (\abs{E} + \abs{V}) + (H-L) \cdot \abs{I}$$
So the optimal vector of reserve prices would give a solution to the maximum
independent set problem.
\end{proof}

Combining the proof of Theorem \ref{thm:lazy_eager_bound} and Proposition
\ref{prop:linear_lazy} we get a $2$-approximation to the optimal revenue
obtained by running an eager second price auction.

We remark that the algorithm proposed in Proposition \ref{prop:linear_lazy} is
different from the heuristic proposed by Hartline and Roughgarden in
\cite{HartlineR09}: their heuristic consists in choosing the vector of monopoly
reserve prices $\vc{r}_M$ such that $r^M_{i} 
= \argmax_{r \geq 0} \sum_{t=1}^T r \cdot \vc{1}\{ b_{i,t} \geq r \}$. They observe
that while this is a very good choice when the valuations are independent and
follow the monotone hazard rate condition, this can be arbitrarily far from the
revenue of the optimal auction if either condition is violated.
We complement the observation showing that this vector can also be arbitrarily
far from the optimal revenue of the second price auction with personalized
reserves (which is a weaker benchmark than the optimal auction).

\begin{claim}
  For every constant $C$, there is a distribution of bids such that
  $\rev(\vc{b}; \vc{r}^M) \leq \frac{1}{C} \rev(\vc{b}; \vc{r}^*)$ for both the
  eager and lazy auctions.
\end{claim}

\begin{proof}
  Consider $2$ bidders and consider a distribution on bids that sets 
   $b_1 = b_2 = 2^{k-1}$ with probability $\frac{1}{2^k}$
  probability, for $k = 1, \ldots, K-1$, and $b_1 = b_2 =
  2^{K-1}+\epsilon$ with probability $\frac{2}{2^K}$. The monopoly reserve
  prices are $r^M_u = 2^{K-1}+\epsilon$ generating revenue $1+\epsilon \cdot
  2^{1-K}$. The optimal vector of reserves is zero for both eager and lazy
  reserves: $\rev(\vc{b};\vc{0}) = K + \epsilon \cdot 2^{1-K}$.
\end{proof}

\begin{figure*}[tb!]
\label{fig:uniform}
\centering
\includegraphics{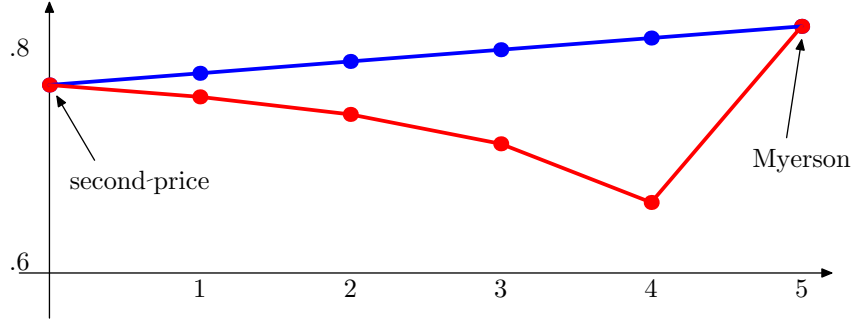}
\caption{$\rev_E(k)$ in red and $\rev_L(k)$ in blue for $n = 5$ agents with iid uniform distributions}
\end{figure*}

\section{A/B Testing}
\label{sec:ab}
In the previous sections we described  auctions with eager and lazy reserve prices,
 and compared the differences in welfare, revenue, and computational complexity 
 of the two approaches. The theoretical models are clean and elegant, and allow us 
 to abstractly reason about the benefits of one approach over another; however they do not 
 capture the messy realities of setting reserve prices in practice. 
 
Before releasing a new model into the wild, a key step is measuring the impact of its change
on a small sample of traffic. This too is a non-trivial step, and requires a lot of care both in setting up 
the experiments~\cite{KohaviKDD07} and effectively measuring long term, as well as, short term impacts~\cite{TangKDD15}, 

%In previous sections we discussed techniques for computing vectors of
%personalized reserve prices and discussed different ways in which one can
%apply them to buyers: either in lazy or eager fashion. The next step is to test
%reserves to verify if the expected lift in revenue will be materialized in live
%traffic. One way to test is to apply vectors of reserves to all buyers in a small
%fraction of queries. This has the effect of verifying if the bid distributions
%were well estimated, but doesn't capture possible reactions that buyers might
%have to the new reserves. While a second price auction with reserves is still
%truthful, buyers could, for example, react to higher prices by buying more
%inventory from other sources reducing their spent -- given that most buyers bid
%not only on one, but on multiple ad exchanges. Other side effect of reserve
%prices is that they can cause budget effects to kick-in earlier, making the
%revenue lift not as high as one would expect. For those and other reasons, it is
%important to run tests where a small subset of buyers are opt-in to personalized
%reserves, while other buyers aren't. Then one is able to evaluate the reaction
%of those buyers to reserves.

One way to test reserve prices is to partition all of the auctions into treatment 
and control, and then only apply reserves in treatment scenarios. However, such a test 
gives a biased estimate of the revenue lift---since the effect on an individual buyer 
is small (as the test applies only to a small fraction of the auctions),  the buyer is unlikely 
to react strategically, for example changing her bidding behavior, or looking for alternative 
places to buy the impressions. 

A different approach is to partition the bidders into treatment and control groups, and only 
apply reserve prices to the bidders in treatment. It is easy to see that applying reserve prices to only a fraction 
of the bidders  will reduce the overall revenue gains, however since a single buyer 
is now subject to reserve prices on all of her auctions, she is more likely to update her 
behavior in response.  

It is natural to expect that applying a vector of optimal reserves to a
subset of the buyers would yield an improvement in revenue that would allow us
to evaluate the impact of applying reserves to all of the buyers.
Counterintuitively, we show that even in the simplest possible setting (iid buyers
with regular distributions), applying the Myerson reserve price to a subset of the bidders 
is {\em worse than not applying any reserve prices at all}!  In other words, as we add bidders 
to the treatment group, the total revenue {\em decreases}, and it is only when all of the bidders
are treated that we realize the revenue gains.  As we show in
Section \ref{sec:empirical}, this phenomenon is not a purely theoretical construct, 
but is also not uncommon in practice.

We state the results formally. 

\begin{theorem}\label{thm:revenue-testing-eager}
  Assume agents are iid with regular distribution $F$ and let 
  $\rev_E(k)$ be the revenue obtained from applying the Myerson reserve price eagerly to
  $k$ out of $n$ agents and applying no reserve to the remaining agents. Then:
  $$\rev_E(0) \geq \rev_E(1) \geq \rev_E(2) \geq \hdots \geq
  \rev_E(n-1)$$
\end{theorem}

\begin{proof}
 Since
  agents are iid we can sample agents according to the following procedure:
  draw $n$ iid samples from $F$ and take a random assignment from them to
  agents. This is equivalent to drawing the valuation $v_i$ for each
  agent and then choosing $k$ out of $n$ at random players to apply reserve
  prices. We will denote by $v^{(t)}$ the $t$-th largest bid.
  
  We consider two scenarios: (i) the highest bidder is above
  the reserve. In this case, she will be allocated the good. (ii) the
  highest bidder is below the reserve. In this case, the agent to whom we 
  allocate the good is the highest agent for whom no reserve price is
  applied. 
  
  Let 
  $\textsc{MaxR}_t(z_1, z_2, \hdots, z_n)$ be the expected maximum of a
  uniformly random subset of $t$ elements drawn from $\{z_1, z_2, \hdots, z_n\}$. Notice that for any
  vector $\vc{z}$, the function $\textsc{MaxR}_t(z_1, z_2, \hdots, z_n)$  is monotone
  non-decreasing in $t$. Second, observe that if $z_i$ are drawn iid, then:
  $$\E[ \max\{z_1, \hdots, z_t\}] = \textsc{MaxR}_t(z_1, z_2, \hdots, z_n).$$ 
  
  For $k<n$, by the Myerson Lemma:
  \begin{equation}\label{eq:revenue}\tag{$\star$}
    \begin{aligned}
    &\rev_E(k) = \E[\phi(v^{(1)}) \cdot \vc{1} \{ v^{(1)} \geq r \}] + \\ & \quad
      \E[\textsc{MaxR}_{n-k}(\phi(v_1), \hdots, \phi(v_n)) \vert v^{(1)} < r]
    \cdot \Pr(v^{(1)} < r)    \end{aligned}
  \end{equation}

  We finish by noticing that when $v_i < r$ then $\phi(v_i) < 0$. Therefore
  the second term above is negative and non-increasing in $k$ by the monotonicity
  of the $\textsc{MaxR}$ operator. For $k=n$, however, the second term
  disappears and we recover the optimal auction.
\end{proof}

We now show that the auction with lazy reserve prices does not suffer from this kind of paradoxical behavior. 

\begin{theorem}\label{thm:revenue-testing-lazy}
  Assume agents are iid with regular distribution $F$ and let $\rev_L(k)$ be the revenue obtained from lazily applying the Myerson reserve price to
  $k$ out of $n$ agents and applying no reserve to the remaining agents. Then
  $$\textstyle\rev_L(k) = \frac{k}{n} \rev_L(n) + \left(1-\frac{k}{n} \right) \rev_L(0)$$
\end{theorem}

\begin{proof}
  As in the previous theorem, we sample the agents by drawing  $n$ iid samples from $F$ and taking a random assignment from them to
  agents. 
  
  In the second price auction with lazy reserves, we will choose the
  agent with bid $v^{(1)}$ and declare her as the winner if $v^{(1)}$ is at
  least the reserve price, which is $r$ with probability $k/n$ and $0$ with the
  remaining probability.  By Myerson's lemma:

  $$\rev_L(k)= \E[\phi(v^{(1)}) \cdot \vc{1} \{ v^{(1)} \geq r^{(1)} \} ]$$
  where $r^{(t)}$ denotes the personalized reserve of agent with the $t$-th
  highest bid. Therefore:

  $$\begin{aligned}
  & \rev_L(k) \\ & =  {\textstyle \left(1-\frac{k}{n} \right)} \E[\phi(v^{(1)})] + 
  {\textstyle \left(1-\frac{k}{n} \right)} \E[\phi(v^{(1)}) \cdot \vc{1} \{ v^{(1)}
  \geq r \}]  \\ & = {\textstyle \left(1-\frac{k}{n} \right)} \rev_L(0) + {\textstyle
  \left(1-\frac{k}{n} \right)} \rev_L(n)
  \end{aligned}$$
\end{proof}

To better understand the detrimental effect of eager reserve prices, it's useful to write equation (\ref{eq:revenue}) in explicit form so that we can
evaluate the impact of applying reserves to a subset of buyers for particular
distributions:

\begin{lemma}
  In the setting of Theorem \ref{thm:revenue-testing-eager}:
  $$\begin{aligned}
  &\rev_E(k) = \int_r^1 \phi(x) \cdot n F(x)^{n-1} f(x)  dx\\ & \quad - \vc{1}\{k<n\}
  F^n(r) \int_0^r \left(\frac{F(x)}{F(r)}\right)^{n-k} \phi'(x) dx
  \end{aligned}$$
\end{lemma}

\begin{proof}
  For the first term, notice that the distribution of $v^{(1)}$ is given by
  $n F(x)^{n-1} f(x)$ since $\Pr(v^{(1)}   \leq x) = \prod \Pr(v_i \leq x) = F(x)^n$.
  For the second term, since the agents are iid, the maximum of $n-k$ randomly chosen agents from 
  among $n$ buyers is identical to the maximum of $n-k$ buyers. By
  the principle of deferred decisions, we can first sample $n-k$ buyers and then
  draw their values. 
  
  Since in the second expression we condition their value to
  be at most $r$, we can simply compute the maximum over $n-k$ buyers with
  density $f(x) / F(r)$ for $0 \leq x \leq r$. By the same argument as before
  the density of the maximum is $(n-k) \cdot \frac{f(x)}{F(r)}\cdot \left(
  \frac{F(x)}{F(r)} \right)^{n-k-1}$. Using this fact, We can write the second
  term as:
  $$\begin{aligned}
    & F(r)^n \cdot \int_0^r \left[ (n-k) \cdot \frac{f(x)}{F(r)}\cdot \left(
    \frac{F(x)}{F(r)} \right)^{n-k-1} \right] \phi(x) dx \\ & \quad =
    \vc{1} \{k<n\} \cdot F^n(r) \left\{ \phi(x) \cdot
    \left(\frac{F(x)}{F(r)}\right)^{n-k} \bigg|_0^r  \right.
  \\ & \quad \quad \quad - \left. \int_0^r \left(\frac{F(x)}{F(r)}\right)^{n-k} \phi'(x) dx \right\} \\ & \quad
      = - \vc{1}\{k<n\} \cdot F^n(r) \int_0^r \left(\frac{F(x)}{F(r)}\right)^{n-k} \phi'(x) dx
  \end{aligned}$$
\end{proof}

\subsection{Case study: Uniform distribution}

To get some intuition about Theorem \ref{thm:revenue-testing-eager} we look at the
setting with $n$ iid bidders distributed according to the $[0,1]$-uniform
distribution, for which $\phi(x) = 2x-1$ and $r = 1/2$:
$$\begin{aligned}
&\rev_E(k) = \int_{1/2}^1 (2x-1) n x^{n-1} dx \\ & \quad - \vc{1}\{k<n\} \cdot \left(\frac{1}{2}\right)^{n}
\cdot \int_0^{1/2} (2x)^{n-k} \cdot 2 dx \\ & = \frac{n+2^{-n}-1}{n+1} -
\vc{1}\{k<n\} \cdot \frac{2^{-n}}{n-k+1}
\end{aligned}$$

In Figure \ref{fig:uniform} we plot the revenue values for $\rev_E(k)$ and $\rev_L(k)$ with $n=5$.  

\section{Experimental Results}\label{sec:empirical}

\begin{figure}
\label{fig:empirical}
\centering
\includegraphics[width=0.8\textwidth]{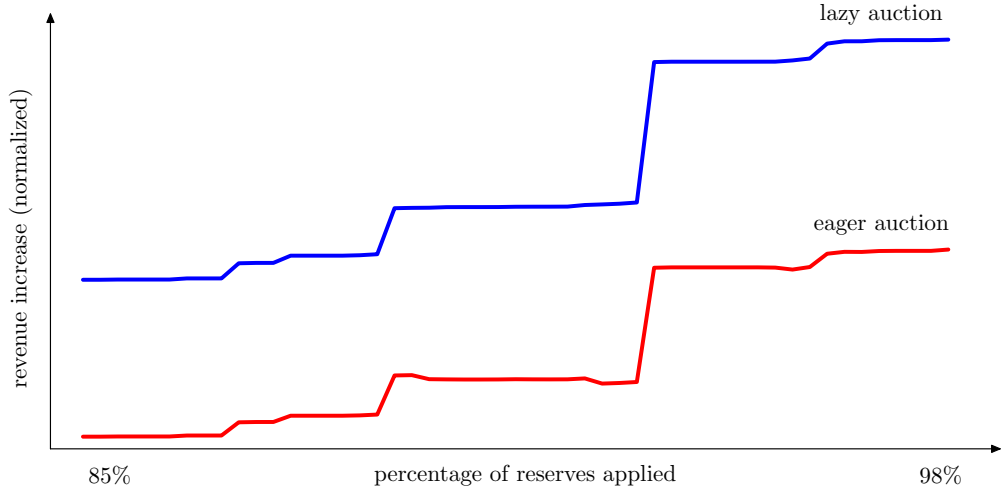}
\caption{Revenue gain obtained by applying the optimal lazy reserve prices to a percentage of the buyers in the auction. Note that the revenue of the lazy auction is monotone non-decreasing, while the eager auction doesn't have this
property.}
\end{figure}

In the previous sections we discussed versions of the second price auction with
personalized reserve prices: the eager (E) and lazy (L) auction, and proved
that while the two notions of reserve prices are incomparable, eager reserves dominate
lazy reserve, except for non-identical and non-independent settings, and in all of the settings total revenue of one 
is always within a factor of two of another. Further, we showed  that applying 
the eager reserve prices to a {\em subset} of the bidders is non monotonic,
and applying reserves to more bidders may (in theory) lead to lower revenues.  In this section we validate these findings by simulating the effect of reserve prices on real world data. 

{\bf Data. } We  collect bids sent to a large advertising exchange over the course of part of a day, 
and then restrict our attention to five ad slots with the highest traffic volume. Each of the ad slots 
has the bids submitted for hundreds of thousands of auctions. We report 
the results for each of the ad slots individually.

\subsection{Non-monotonicity}
Recall that in Section \ref{sec:ab} we considered the standard setting where each 
bidder $i$ draws an independent bid from bid distribution $F_i$. In the case of eager
reserve prices, applying reserves to only a subset of the bidders leads to a {\em decrease} 
in the overall revenue to the auctioneer. 

In this section we show that this scenario is not restricted to theory.  We consider a set of auctions
and compute the optimal lazy reserve price for each bidder. We then apply this reserve price in 
both lazy and eager fashions to a subset of the bidders and plot the overall revenue. The results 
are shown in Figure \ref{fig:empirical}. 

We confirm our theoretical findings: while the revenue in the auction with the lazy reserve pricing strategy 
is monotonically increasing in the number of bidders subject to the reserve, the same cannot be said 
for the greedy reserve pricing strategy.  The overall trend is positive, but it is non monotonic, and at times increasing the {\em number} of bidders
subject to the optimal reserve price decreases the overall revenue. 

\subsection{Revenue Gains}
Our second set of experiments addresses the question whether eager or lazy reserve prices 
lead to higher revenues in practice. For the best comparison, we simulate the auctions under 
both lazy and eager strategies while applying the optimal lazy reserve prices $\vc{r}_L^*$ (see Proposition \ref{prop:linear_lazy})
and the monopoly reserve prices $\vc{r}_M$ (see the discussion following the
proof of Proposition \ref{prop:nphard}).

We compute the revenue lifts due to personalized reserves under four strategies: 
$$\Delta_L(\vc{r}^*_L) = \E[\rev_L(\vc{b}; \vc{r}^*_L)] - \E[\rev(\vc{b};
\vc{0})]$$ 
$$\Delta_E(\vc{r}^*_L) = \E[\rev_E(\vc{b}; \vc{r}^*_L)] - \E[\rev(\vc{b};
\vc{0})]$$ 
$$\Delta_L(\vc{r}_M) = \E[\rev_L(\vc{b}; \vc{r}_M)] - \E[\rev(\vc{b};
\vc{0})]$$ 
$$\Delta_E(\vc{r}_M) = \E[\rev_E(\vc{b}; \vc{r}_M)] - \E[\rev(\vc{b};
\vc{0})]$$ 
where the expectation $\E[\cdot]$ denotes the average over all historical
queries.

We normalize the lifts  by setting the revenue of the lazy reserve auction with
optimally set reserve prices to $1$ ($\Delta_L(\vc{r}^*_L) = 1$
), and report multiplicative improvement over this setting. 

\begin{table}
\centering
\caption{Revenue Lift Comparison. The units are normalized so that
$\Delta_L( \vc{r}^*_L) = 1$}
\label{table:revlift}
\vspace{0.1in}
\begin{tabular}{|c|c|c|} \hline
$\Delta_E( \vc{r}^*_L)$ &
$\Delta_L( \vc{r}_M)$ &
$\Delta_E( \vc{r}_M)$ \\ \hline
1.13204 & 0.892116 & 1.18965 \\
1.24867 & 0.958164 & 1.28977 \\
1.16233 & 0.942408 & 1.09623 \\
1.19286 & 0.886347 & 1.11872 \\
1.14805 & 0.942208 & 1.08097\\
\hline\end{tabular}
\end{table}

Table \ref{table:revlift} shows that the eager auction outperforms the lazy
auction in practice both when the optimal lazy reserves are used and when
monopoly reserves are used. We also observe that there is no clear winner
between using optimal lazy reserves or monopoly reserves in the eager auction.
For $3$ out of $5$ slots, the vector of optimal lazy reserves outperforms the
monopoly reserves. In practice one may want to start with any of those two
vectors and perform local updates to improve the performance of the eager auction.

Another axis along which it makes sense to compare the eager and lazy auctions
is the welfare loss of the allocation, i.e., how much welfare is lost due to the
application of reserve prices. If $\textsc{W}^L(\vc{b}; \vc{r})$ and
$\textsc{W}^E(\vc{b}; \vc{r})$ are respectively the welfare of the lazy and eager auctions
when reserve price $\vc{r}$ are applied, we define the quantities:
$$\tilde{\Delta}_L(\vc{r}^*_L) =  \E[\textsc{W}(\vc{b}; \vc{0})] -
\E[\textsc{W}_L(\vc{b}; \vc{r}^*_L)] $$ 
$$\tilde{\Delta}_E(\vc{r}^*_L) = \E[\textsc{W}(\vc{b}; \vc{0})] -
\E[\textsc{W}_E(\vc{b}; \vc{r}^*_L)] $$ 
$$\tilde{\Delta}_L(\vc{r}_M) = \E[\textsc{W}(\vc{b}; \vc{0})] -
\E[\textsc{W}_L(\vc{b}; \vc{r}_M)]$$ 
$$\tilde{\Delta}_E(\vc{r}_M) = \E[\textsc{W}(\vc{b}; \vc{0})] -
\E[\textsc{W}_E(\vc{b}; \vc{r}_M)]$$ 
which are the analogues for social welfare of the quantities described in Table
\ref{table:revlift}.

\begin{table}
\centering
\caption{Welfare Loss Comparison. The units are normalized so that
$\tilde{\Delta}_L( \vc{r}^*_L) = 1$}
\label{table:welflift}
\vspace{0.1in}
\begin{tabular}{|c|c|c|} \hline
$\tilde{\Delta}_E( \vc{r}^*_L)$ &
$\tilde{\Delta}_L( \vc{r}_M)$ &
$\tilde{\Delta}_E( \vc{r}_M)$ \\ \hline
0.689856 & 1.34808 & 0.896388 \\
0.70264 & 1.16671 & 0.814005 \\
0.602512 & 1.046 & 0.695923 \\
0.616777 & 1.0458 & 0.647172 \\
0.636252 & 0.995548 & 0.671452 \\
\hline\end{tabular}
\end{table}

As expected, Table \ref{table:welflift} shows that the welfare loss in the eager
auction is larger than the welfare loss in the lazy auction. Similarly to what
occurs for revenue, there is no clear winner between the optimal lazy reserves
and the monopoly reserves for the eager auction.

\section{Conclusion}
%In this work we provided a more complete characteristics of the properties of lazy and eager auctions.  We showed that lazy reserve prices are easy to optimize, but lead to lower welfare, and often lower revenues than the eager reserve prices. We investigated the effect of applying reserves to a subset of bidders, as a part of testing their effects in production, and proved that eager reserves can paradoxically {\em decrease} the total revenue when applied to a group of bidders, and observed the same effect in simulations on real world data. 

The results in this work follow two major themes. The first lies in devising methods for testing new strategies, be they reserve prices, bidding agents, or mechanisms, in competitive environments~\cite{ABTesting}. As we saw in Section \ref{sec:ab} these can have counterintuitive effects, and robust experimental design and analysis is required for making correct decisions. 

The second is in foregoing the assumption that bidders' value distributions are known and investigating the computational complexity of optimizing from previously observed data. Whether designing new mechanisms in these environments~\cite{ColeR}, understanding the amount of past data needed~\cite{MorgensternR}, or, providing approximation algorithms, as we do in this work, this is a rich and exciting open area.  

\bibliographystyle{alpha}
%\balancecolumns
\bibliography{biblio}

\end{document}